\documentclass[11pt]{article}    
\usepackage{amsmath,amssymb,amsthm,amsfonts,natbib}

\title{Probability matrices, non-negative rank, and parameterizations of mixture models}

\author{Enrico Carlini \\ Department of Mathematics, Politecnico di Torino \\
Corso Duca degli Abruzzi, 24 \\
10129 Turin, Italy \\
\tt{enrico.carlini@polito.it}
\and Fabio Rapallo \\ Department DISTA, University of Eastern Piedmont \\
Viale Teresa Michel, 11 \\
15121 Alessandria, Italy \\ \tt{fabio.rapallo@mfn.unipmn.it} }

\begin{document}

\newtheorem{thm}{Theorem}[section]
\newtheorem{definition}[thm]{Definition}
\newtheorem{proposition}[thm]{Proposition}
\newtheorem{lemma}[thm]{Lemma}
\newtheorem{corollary}[thm]{Corollary}
\newtheorem{remark}[thm]{Remark}
\newtheorem{example}[thm]{Example}

\maketitle                 

\begin{abstract}
In this paper we parameterize non-negative matrices of sum one and
rank at most two. More precisely, we give a family of
parameterizations using the least possible number of parameters.
We also show how these parameterizations relate to a class of
statistical models, known in Probability and Statistics as mixture
models for contingency tables.
\bigskip


{\it AMS 2000:} 15A51, 62H17.

\end{abstract}

\section{Introduction} \label{introsect}

The study of non-negative matrices with fixed rank has recently
attracted a great deal of work both theoretical and applied. One
of the main problems in this field is the so-called ``non-negative
matrix factorization problem'', which can be shortly stated as
follows. Given a non-negative matrix $A \in {\mathbb R}_+^{I
\times J}$ (where ${\mathbb R}_+$ denotes the set of real
non-negative numbers), one has to find an approximation of $A$ as
a linear combination of $k$ dyadic products $c_ir_i^t$, where the
$c_i$'s and $r_i$'s are vectors with non-negative entries, i.e.
$c_i \in {\mathbb R}_+^J$ and $r_i \in {\mathbb R}_+^I$.

The rank of a matrix gives the numbers of rank one matrices, i.e.
dyadic products, needed to write the matrix as a sum of dyads. But
there are no non-negative conditions on the vectors of the dyads.
The non-negativity constraints make the situation more complex and
one has to work with the {\it non-negative rank} of the matrix
(see e.g. \cite{cohen|rothblum:93}), which is in general bigger
than the ordinary rank. Therefore, it is not possible in general
to decompose a rank $k$ matrix into the sum of exactly $k$ dyadic
products $c_ir_i^t$ where $c_i$ and $r_i$ are non-negative
vectors. We will review the main results about non-negative rank
in the next section.

In recent literature, a number of results and algorithms for
non-negative matrix factorization have been published, see e.g.
\cite{lee|seung:00}. In \cite{catral|han|neumann|plemmons:04}
special techniques for symmetric tables are presented, while in
\cite{ho|vandooren:08} the case of fixed row and column sums is
analyzed, with applications to stochastic matrices. In
\cite{finesso|spreij:06}, the authors discuss some connections
between the factorization problem and the notion of
$I$-divergence, which has a well known statistical role, see e.g.
\cite{dacunhacastelle|duflo:86} and \cite{pardo:05}.

From the point of view of Probability, non-negative matrices are a
natural tool in the analysis of two-way contingency tables. A
two-way contingency table $A=(a_{i,j})$ collects data from two
categorical random variables measured on $n$ subjects. Let us
suppose that the first variable $X$ has $I$ levels $1, \ldots, I$
and the second variable $Y$ has $J$ levels $1, \ldots ,J$. The
element $a_{i,j}$ is the count of subjects with $X=i$ and $Y=j$.
Therefore, $A$ is an $I \times J$ matrix with non-negative integer
entries.

A joint probability distribution for the pair $(X,Y)$ is a {\it
probability matrix} with $I$ rows and $J$ columns $P=(p_{i,j})$ of
non-negative real numbers such that $\sum_{i,j} p_{i,j} =1$. A
statistical model ${\mathcal M}$ for $I \times J$ contingency
tables is a set of probability distributions, i.e. a subset of the
simplex
\begin{equation}
\Delta = \left\{ P=(p_{i,j}) \ : \ p_{i,j} \geq 0, \ \sum_{i,j}
p_{i,j} =1  \right\}\subset\mathbb{R}_{+}^{I\times J} \, .
\end{equation}
One of the most widely used models for two-way contingency tables
is the independence model, see e.g. \cite{agresti:02}. It is
defined through the vanishing of all $2 \times 2$ minors of the
generic matrix, i.e. by the equations
\begin{equation}
p_{i,j}p_{l,h} - p_{i,h}p_{l,j} = 0 \ \ \ \ 1 \leq i < l \leq I, \
1 \leq j < h \leq J \, ;
\end{equation}
thus, the points of the independence model are rank $1$ matrices.

Recent developments in Statistics have shown the relevance of
probability models whose points are matrices of rank at most $2$.
One example in this direction, based on a special symmetric
matrix, is the so-called ``$100$ Swiss francs problem'', see
\cite{sturmfels:08}. This problem comes from Computational
Biology, where it is useful to analyze the alignment of DNA
sequences, see \cite{pachter|sturmfels:05}. Although this
particular problem has been solved in \cite{gao|jiang|zhu:08}, the
study of fixed-rank probability matrices is mainly unexplored.

As the sum of $k$ matrices with rank $1$ has rank at most $k$, the
matrices which can be written as the sum of $k$ dyadic products
encode the notion of mixture of $k$ distributions from
independence models.

In Probability and Statistics it is interesting not only to study
the approximation problem mentioned above, but also to have a
parametrization of the models. While for rank $1$ matrices the
parametrization is easy, see e.g. \cite{agresti:02}, the problem
becomes difficult in the case of higher non-negative ranks.
Already for $k=2$, in \cite{fienberg|hersh|rinaldo|zhou:10} it is
shown that the model is not identifiable, meaning that different
parameter values lead to the same probability distribution.

This issue is a well known problem in statistical modelling called
``parameter redundancy'', see \cite{catchpole|morgan:97} and
\cite{catchpole|morgan|freeman:98}. The detection of parameter
redundancy has a major relevance in maximum likelihood estimation,
where the parameters of a statistical models are estimated through
the maximization of a real-valued function called ``likelihood
function'', see e.g. \cite{agresti:02}. In the papers mentioned
above, the authors propose a purely analytical technique to detect
the parameter redundancy of a statistical model, by computing the
rank of the Jacobian matrix of a specific function. The redundancy
is checked through Symbolic Algebra computations and the problem
of redundancy is overcome via additional linear constraints on the
parameters.

In this paper, we propose a method which uses linear algebra to
make the maximization problem simpler by reducing the number of
parameters involved. Then the usual analytic techniques can be
used in a more effective way.

The paper is organized as follows: in Section \ref{theorec} we
introduce some definition an we recall some basic facts. In
Section \ref{expdesign} we study the problem of parameters
redundancy form a geometric point of view. In Section
\ref{appsection} we show a possible application of our results.

\section{Definition and background material} \label{theorec}

Let $P=(p_{i,j})$ be a probability matrix with $I$ rows and $J$
columns, i.e. $P\in\Delta$. In order to simplify the formulae, let
us suppose that $I \leq J$. Let $k$ be an integer, $1 \leq k \leq
I$.

\begin{definition} \label{simplestdef}
A probability matrix $P$ is the mixture of $k$ independence models
if it can be written in the form:
\begin{equation} \label{mixturedef}
P= \alpha_1 c_1r_1^t + \ldots + \alpha_k c_kr_k^t
\end{equation}
where for all $h=1, \ldots , k$
\begin{itemize}
\item $\alpha_h \in \mathbb{R}_+$ and $\sum_h \alpha_h = 1$;

\item $r_h\in\mathbb{R}_+^J$ and $\sum_i r_h(i) = 1$;

\item $c_h\in\mathbb{R}_+^I$ and $\sum_j c_h(j) = 1$.
\end{itemize}
\end{definition}

Definition \ref{simplestdef} contains a simple parametric form of
the probability distribution which has an intuitive probabilistic
counterpart. Let us suppose that we have $k$ pairs of dice, say
$(D_{1,r},D_{1,c}), \ldots , (D_{k,r},D_{k,c})$, where $D_{h,r}$
has $J$ facets and distribution $r_h$ and $D_{h,c}$ has $I$ facets
and distribution $c_h$. We choose a pair of dice with probability
distribution $\alpha=(\alpha_1, \ldots, \alpha_k)$ and we roll the
selected pair of dice. The resulting distribution is just a
mixture distribution as in Eq. \eqref{mixturedef}.

As a Linear Algebra counterpart, the definition above is strictly
related with the notion of non-negative rank of a matrix. For more
on non-negative rank see, e.g., \cite{cohen|rothblum:93}. We
recall here some useful facts.

\begin{definition}  \label{posrank}
Given a matrix $P$ with real non-negative elements, the
non-negative rank of $P$ is the smallest number of non-negative
column vectors $v_1, \ldots, v_k$ of $P$ such that each column of
$P$ has a representation as a linear combination of $v_1, \ldots,
v_k$ with non-negative coefficients. The non-negative rank of a
matrix $P$ is denoted with $\mathrm{rk}_+(P)$.
\end{definition}

The definition above has an equivalent formulation in terms of
linear combinations of row vectors. In the following proposition
we summarize the main properties of the non-negative rank. The
reader can refer to \cite{cohen|rothblum:93} for proofs and
further details. The non-negative rank is of special relevance for
Probability and Statistics. In fact, $\mathrm{rk}_+(A)$ is the
number of dyadic products of non-negative vectors that we can use
to represent $A$.

\begin{proposition} \label{mainproperties}
Let $P$, $Q$ be two non-negative matrices with $I$ rows and $J$
columns.
\begin{itemize}
\item[(a)] $\mathrm{rk}(P) \leq \mathrm{rk}_+(P) \leq
\min\{I,J\}$;

\item[(b)] $\mathrm{rk}_+(P) = \mathrm{rk}_+(P^t)$;

\item[(c)] $\mathrm{rk}_+(P+Q) \leq \mathrm{rk}_+(P) +
\mathrm{rk}_+(Q)$.
\end{itemize}
Moreover, if $P$ has dimensions $I \times K$ and $Q$ has
dimensions $K \times J$, then $\mathrm{rk}_+(PQ) \leq
\min\{\mathrm{rk}_+(P),\mathrm{rk}_+(Q)\}$.
\end{proposition}

Items $(b)-(d)$ in Proposition \ref{mainproperties} show that the
non-negative rank has properties similar to the classical rank. In
general, the rank and non-negative rank are different, as shown by
the following matrix
\begin{equation*}
\begin{pmatrix} 1 & 0 & 1 & 0 \\
                          1 & 0 & 0 & 1 \\
                          0 & 1 & 1 & 0 \\
                          0 & 1 & 0 & 1
                          \end{pmatrix}
\end{equation*}
which has rank $3$ but non-negative rank $4$.

Among the cases where the rank and the non-negative rank coincide,
there are the following special classes of matrices
\cite{cohen|rothblum:93}.

\begin{proposition} \label{luckycases}
Let $P$ be a non-negative matrix with $I$ rows and $J$ columns.
\begin{itemize}
\item[(a)] If $\mathrm{rk}(P) \leq 2$ then $\mathrm{rk}_+(P) =
\mathrm{rk}(P)$;

\item[(b)] If $P$ is diagonal, then $\mathrm{rk}_+(P) =
\mathrm{rk}(P)$.
\end{itemize}
\end{proposition}

In what follows we will heavily use part (a) of Proposition
\ref{luckycases}. Hence, for the convenience of the reader, we
produce a self contained proof of this fact for probability
matrices.

\begin{lemma}\label{nonnegativelincomb}
Let $P$ be a probability matrix. If $\mathrm{rk}(P)\leq 2$, then
$\mathrm{rk}_+(P)=\mathrm{rk}(P)$.
\end{lemma}
\begin{proof}
If $\mathrm{rk}(P)=1$ then the proof is trivial; thus we will
assume $\mathrm{rk}(P)=2$. Denote with $C_i,i=1\ldots,J$ the
columns of $P$. We will show that there exist two columns, say
$\bar C$ and $\tilde C$, such that $C_i=t_i\bar C+s_i\tilde C$ for
all $i$ and the coefficients $t_i$'s and $s_i$'s are non-negative.

Clearly, as $P$ has rank at most two, all columns are linear
combinations of two fixed ones. Without loss of generality, we may
assume that $C_1$ and $C_2$ are linearly independent. Thus for any
other column we have $C_i=t_iC_1+s_iC_2$. If all the pairs
$(t_i,s_i)$ are non-negative we are done. Otherwise, consider in
the plane $\mathbb{R}^2$ the rays spanned by the pairs $(t_i,s_i)$
and let $(\bar t,\bar s)$ and $(\tilde t,\tilde s)$ be the
extremal rays and denote by $\bar C$ and $\tilde C$ the
corresponding columns. We recall that the extremal rays are the
minimal generators of the convex cone spanned by the the pairs
$(t_i,s_i)$. Now consider the angle $\phi$ between the extremal
rays containing at least one positive semi-axis. If $\phi<\pi$
radiants then we are done by using the addition rule for vectors
in the plane and all the columns are non-negative linear
combinations of $\bar C$ and $\tilde C$. If $\phi=\pi$ radiants we
get the contradiction as $\bar C+\tilde C=0$ and hence $C_1$ and
$C_2$ would be proportional. If $\phi>\pi$ we get again a
contradiction. In fact, a non-negative combination of the extreme
rays would be in the negative quadrant. Hence, a non-negative
linear combination of $\bar C$ and $\tilde C$ would be
non-positive and hence equal to zero being $P$ non-negative. Thus,
$C_1$ and $C_2$ would be proportional again.
\end{proof}

\section{Parameters and parameterizations} \label{expdesign}

Often in Probability and in Statistic models are described using
parameters. This description can be easily expressed in geometric
terms. Given the variety representing the model we look for a
surjective function into it. More precisely, if ${\mathcal M}$ is
the model, a surjective function
$U\subseteq\mathbb{R}^n\longrightarrow {\mathcal M}$ gives a {\it
parametrization} of ${\mathcal M}$. If the function we found is
described by rational functions and its image is dense in the
model, we say that the map is {\it dominant} and we describe the
model up to a measure zero set.

Given a model ${\mathcal M}$ there are two basic questions: Does
there exist a dominant map $\mathbb{R}^n\longrightarrow {\mathcal
M}$? What is the smallest $n$ for which such a map exists?
Answering the first question is a deep and difficult problem in
Geometry called ``the unirationality problem'', see \cite[page
87]{harris:92}. The second question is difficult too, but we can
easily give a bound on $n$ using the dimension of ${\mathcal M}$,
namely we must have $n\geq \dim {\mathcal M}$.

When we have a parametrization of a model ${\mathcal M}$ such that
$n=\dim {\mathcal M}$ we say that the parametrization is {\it
non-redundant}, or that the parameters are non-redundant. It is
not always possible to find a non-redundant parametrization. But,
in some interesting situations, it is possible to decompose the
model ${\mathcal M}$ as union of subvarieties and for each of this
one can find a non-redundant parametrization. We will give
examples of these phenomena in the case of rank $k$ and rank 2
mixture models.

\subsection{A parametrization for the rank $k$ matrices}

Given natural numbers $I\leq J$ we consider the following family
of matrices with rank at most $k$:
\begin{equation*}
\mathcal{M}_k=\left\{ P=(p_{i,j}) \in \mathbb{R}^{I\times J} :
\mathrm{rk}(P)\leq k  \ , \ \sum_{i,j}p_{i,j}=1 \right\} \, .
\end{equation*}

As the elements of $\mathcal{M}_k$ have rank at most $k$, they can
be written as a linear combination of at most $k$ rank one
probability matrices. More precisely, if $P\in\mathcal{M}_k$ then
\begin{equation}\label{parameterization1}
P = \alpha_1 c_1r_1^t + \ldots + \alpha_k c_kr_k^t
\end{equation}
for a choice of scalars $\alpha_i's$ and of column vectors $c_i$'s
and $r_i$'s. Hence, we can represent elements of $\mathcal{M}_k$
using
\begin{equation*}
k(I+J)+k
\end{equation*}
parameters.  In other words, \eqref{parameterization1} gives a
surjective polynomial map
\begin{equation*}
\mathbb{R}^{k(I+J)+k}\longrightarrow \mathcal{M}_k \, .
\end{equation*}

We recall that a map between algebraic varieties, say
$V_1\longrightarrow V_2$, can be a parametrization, only if $\dim
V_1 \geq \dim V_2$. To know whether the parameters we are using
are necessary or redundant, we need to know the dimension of
$\mathcal{M}_k$ and compare it with ${k(I+J)+k}$.

\begin{proposition}\label{dimPROP}
With the notation above, we have
\begin{equation*}
\dim \mathcal{M}_k \leq k(I+J)-k^2-1 \, .
\end{equation*}
\end{proposition}
\begin{proof}
The dimension of the family of complex $I\times J$ matrices of
rank at most $k$ is well known to be $ k(I+J)-k^2$, see
\cite{harris:92}. Imposing that the sum of al the entries is $1$
and taking real matrices give the bound.
\end{proof}

Proposition \ref{dimPROP} shows that the parametrization
\eqref{parameterization1} is redundant and we are using more
parameters than the best possible value. Actually, it is not
possible to use $k(r+s)-k^2-1$ parameters to get all the elements
of $\mathcal{M}_k$. In the case of $k=2$ we will show how to
decompose $\mathcal{M}_k$ in open subsets which can each be
described using the optimal number of parameters.

\subsection{Non-redundant parameterizations of probability models for $k=2$}

In this section we only deal with matrices of rank at most two.
Hence we fix $k=2$ and we set
\begin{equation*}
\mathcal{M}=\mathcal{M}^+_2=\left\{P \in {\mathbb R}_+^{I \times
J} \ , \ \mathrm{rk}_+(P) = 2 \right\} \cap \Delta \, .
\end{equation*}
In this situation, $\dim \mathcal{M}\leq 2I+2J-5$ and we will use
this number of parameters to describe $\mathcal{M}$, hence finding
a non-redundant parametrization. Set $D=2I+2J-5$. We will
construct maps
\begin{equation*}
f_{j_1,j_2}:U_{j_1,j_2}\subset\mathbb{R}^D\longrightarrow\mathcal{M}
\end{equation*}
for $1\leq j_1 < j_2 \leq J$, with the property that the union of
the images of the $f_{j_1,j_2}$ is the whole $\mathcal{M}$, i.e.
$\bigcup\mathrm{Im}(f_{j_1,j_2})=\mathcal{M}$.

Each map is constructed in such a way that
$\mathrm{Im}(f_{j_1,j_2})$ is contained in the open subset of the
matrices with the $j_1$-th and the $j_2$-th columns linearly
independent. We give an explicit description only for $f_{1,2}$,
the other cases being completely analogous. We set
\begin{equation*}
f_{1,2}(a_1,\ldots,a_{I-1},b_3,\ldots,b_{J},c_1,\ldots,c_{I-1},d_3,\ldots,d_{J},\alpha)=
\end{equation*}
\begin{equation*}
=\alpha \left(\begin{array}{c} a_1 \\ a_2 \\ \vdots \\ a_{I-1} \\
1-\sum a_i\end{array}\right) \left(\begin{array}{ccccc}1-\sum b_i
& 0 & b_3 & \ldots & b_J\end{array}\right)+
\end{equation*}
\begin{equation*}
+(1-\alpha)\left(\begin{array}{c} c_1 \\ c_2 \\ \vdots \\
c_{I-1}
\\ 1-\sum c_i\end{array}\right) \left(\begin{array}{ccccc}0 & 1-\sum
d_i &  d_3 & \ldots & d_J\end{array}\right),
\end{equation*}
defined on
\begin{equation*}
U'_{1,2}=\left\lbrace
(a_1,\ldots,a_{I-1},b_3,\ldots,b_{J},c_1,\ldots,c_{I-1},d_3,\ldots,d_{J},\alpha)\in\mathbb{R}^
D : \right.
\end{equation*}
\begin{equation*}
\left. 0\leq a_i,b_i,c_i,d_i,\alpha\leq 1 \ \mbox{ and } \ 0\leq
\sum a_i,\sum b_i,\sum c_i,\sum d_i\leq 1 \right\rbrace.
\end{equation*}

To define $f_{j_1,j_2}$ one simply moves element in the row
vectors. In the first row vector the $1-\sum b_i$ element is moved
in position $j_1$ and the $0$ is moved in position $j_2$;
similarly for the second row vector.

\begin{remark} With standard computations one can easily check that
\begin{equation*}
\mathrm{Im}(f_{j_1,j_2}) \subset \mathcal{M}
\end{equation*}
for all $j_1$ and $j_2$, $j_1<j_2$.
\end{remark}

Now we analyze the functions $f_{j_1,j_2}$ in order to derive some
useful properties. We  work with $f_{1,2}$ and all the results
trivially extend to the other functions.

\begin{lemma}\label{inverse}
Let $P\in\mathcal{M}$ be the following matrix
\begin{equation*}
P=\left(
\begin{array}{cccccc}
x_1 & y_ 1 & \ldots & t_ix_1+s_iy_1 & \ldots & t_Jx_1+s_Jy_1\\ \\
\vdots & \vdots & \vdots & \vdots & \vdots\\ \\
x_I & y_I & \ldots & t_ix_I+s_iy_I & \ldots & t_Jx_I+s_Jy_I
\end{array}
\right)
\end{equation*}

where the coefficients $x_i,y_i,s_i$ and $t_i$ are non-negative.

If the first two columns of $P$ are non-zero, we set
\begin{equation*}
a_i={\frac{x_i} {\sum x_i}},c_i={\frac {y_i} {\sum
y_i}},b_i={\frac {t_i} {1+\sum t_i}},d_i={\frac {s_i} {1+\sum
s_i}}
\end{equation*}
and also $\alpha=(\sum t_i+1)\sum x_i=1-(\sum s_i+1)\sum y_i$.

If $\sum y_i=0$, we set
\begin{equation*}
a_i={\frac {x_i} {\sum x_i}},b_i={\frac{t_i} {1+\sum t_i}}
\end{equation*}
and also $\alpha=1$, and $c_i=d_i=0$ for all $i$.

If $\sum x_i=0$, we set
\begin{equation*}
c_i={\frac{y_i} {\sum y_i}},d_i={\frac{s_i} {1+\sum s_i}}
\end{equation*}
and also $\alpha=0$ and $a_i=b_i=0$ for all $i$.

If we set
$P'=(a_1,\ldots,a_{I-1},b_3,\ldots,b_{J},c_1,\ldots,c_{I-1},d_3,\ldots,d_{J},\alpha)$,
then $P'\in U'_{1,2}$ and
\begin{equation*}
f_{1,2}(P')=P \, .
\end{equation*}
\end{lemma}
\begin{proof}
The definition of $P'$ and the condition on the entries of $P$
yield that $P' \in U'_{1,2}$. A straightforward computation shows
that $f_{1,2}(P')=P$. The two expressions for the parameter
$\alpha$ coincide as $P$ is a matrix with sum one.
\end{proof}

Finally we can show that the maps $f_{j_1,j_2}$ give a
parametrization of $\mathcal{M}$.
\begin{corollary}
The variety $\mathcal{M}$ is covered by the images of the
functions $f_{j_1,j_2}$, more precisely
$\bigcup\mathrm{Im}(f_{j_1,j_2})=\mathcal{M}$.
\end{corollary}
\begin{proof}
Let $P\in\mathcal{M}$, by Lemma \ref{nonnegativelincomb} we know
that $P$ can be written as in the statement of Lemma \ref{inverse}
for some columns $C_{j_1}$ and $C_{j_2}$ and hence
$P\in\mathrm{Im}(f_{j_1,j_2})$.
\end{proof}

\section{An application}\label{appsection}

It is often interesting to find maxima and minima of a function
over a variety. As an example consider the well known likelihood
function. We will use the parametrization we found in the previous
sections to propose a strategy to study extremal points on
$\mathcal{M}$.  The advantage of this approach is that we are
going to study functions involving the least possible number of
variables as the parametrization we found is non-redundant.

\begin{remark}\label{interior2}
Given a function $F:\mathcal{M}\longrightarrow\mathbb{R}$ we
consider the composite functions $F\circ f_{j_1,j_2}$. Consider a
point $P=f_{j_1,j_2}(P')\in\mathcal{M}$ such that $P$ is in the
interior of $\mathrm{Im}(f_{j_1,j_2})$. Then $P$ is a
maximum/minimum for $F$ if and only if $P'$ is a maximum/minimum
for $F\circ f_{j_1,j_2}$.
\end{remark}

Using Remark \ref{interior2} we can apply the usual gradient and
Hessian matrix approach to detect extremal points of $F$ lying in
the interior of one of the $\mathrm{Im}(f_{j_1,j_2})$. Hence it
useful to have the following:

\begin{lemma}\label{interior}
If $P'$ is in the interior of $U'_{j_1,j_2}$ then
$f_{j_1,j_2}(P')$ is in the interior of
$\mathrm{Im}(f_{j_1,j_2})$.
\end{lemma}
\begin{proof}
We produce a proof for $j_1=1$ and $j_2=2$ but a completely
analogous argument works in the general situation. Given $P'$ we
compute $P=f_{1,2}(P')$ and thus we write $P$ as in the statement
of Lemma \ref{inverse}. Moreover, as $P'$ is in the interior of
$U'_{1,2}$ the coefficients $t_i$ and $s_i$ in $P$ are strictly
positive. Now consider a neighborhood $U$ of $P$. Given a matrix
$Q\in U$ we can write it in the form of Lemma \ref{inverse} by
computing the coefficients $t_i$ and the $s_i$. This is done by
solving linear systems of equations having the elements of $Q$ as
coefficients. Hence, it is possible to choose a suitable $U$ such
that for all the matrices in $U$ the coefficients $t_i$ and $s_i$
are strictly positive. In conclusion, the formulae of Lemma
\ref{inverse} produce a map $g_{1,2}:U\longrightarrow U'_{1,2}$.
It is straightforward to see that $g_{1,2}$ is a continuous map on
$U$ and that the map
\begin{equation*}
f_{1,2}\circ g_{1,2}
\end{equation*}
is the identity map. Now we take a neighborhood of $P'$, say
$U'\subset f^{-1}_{1,2}(U)$. Then we get a neighborhood of $P$
\begin{equation*}
g^{-1}_{1,2}(U')\subset\mathrm{Im}(f_{1,2})
\end{equation*}
and we are done.
\end{proof}

Lemma \ref{interior} shows that we only have to worry about points
of $\mathcal{M}$ which are images of boundary points of
$U_{j_1,j_2}$. Thus it is useful to have the following
description:

\begin{lemma}\label{boundary}
Let $P'\in U'_{j_1,j_2}$ be the point
\begin{equation*}
P'=(a_1,\ldots,a_{I-1},b_3,\ldots,b_{J},c_1,\ldots,c_{I-1},d_3,\ldots,d_{J},\alpha)
\end{equation*}
and let $P=f_{j_1,j_2}(P')$. Then the following hold:
\begin{enumerate}
\item\label{case1} if any of the coefficients $a_i$ or $c_i$ is
zero then $P$ is a point of the boundary of $\mathcal{M}$;

\item\label{case2} if $\sum a_i=1$ or $\sum b_i=1$ then $P$ is a
point of the boundary of $\mathcal{M}$;

\item if $\alpha=0$ or $\alpha=1$ then $P$ is a rank one matrix;

\item if any of the coefficients $b_i$ or $d_i$ is zero then is
$P$ has at least two proportional columns;

\item if $\sum a_i=1$ or $\sum b_i=1$ then $P$ has at least two
proportional columns.

\end{enumerate}
\end{lemma}
\begin{proof}
For  \eqref{case1} and \eqref{case2} it is enough to notice that
$P$ has some zero element. Hence a neighborhood of $P$ contains
matrices with negative entries. Thus $P$ is on the boundary of
$\mathcal{M}$. The other cases are obtained by direct
computations.
\end{proof}

By Lemma \ref{boundary} we see that the composite map $F\circ
f_{j_1,j_2}$ will detect maxima and minima of $F$ if these
extremal points do not have rank one or if they have rank two and
do not have two proportional columns. In many situation of
interest rank one matrices can be efficiently treated, e.g. for
the likelihood function. Rank two matrices with proportional
columns can be treated using our parametrization in a subtler way.

\begin{lemma}
Let $P=f_{j_1,j_2}(P'_{j_1,j_2})$ be a rank two matrix with at
least two proportional columns. Then a neighborhood of $P$ in
$\mathcal{M}$ can be covered using images of neighborhoods of
$P'_{j_1,j_2}$ in $U'_{j_1,j_2}$ for different pairs $(j_1,j_2)$.
\end{lemma}
\begin{proof}
Given $P$, choose two independent columns, say the $j_1$-th and
the $j_2$-th. As $P$ has proportional columns, when written as in
Lemma \ref{inverse} some of the coefficients $t_i$ and $s_i$
vanish. Hence, in each neighborhood of $P$ there will be matrices
requiring negative values of the coefficients $t_i$ or $s_i$. Then
there is no neighborhood where the formulae of the Lemma can be
applied to get and inverse on $f_{j_1,j_2}$ and hence we can not
reproduce the argument of Lemma \ref{interior}. But we can find a
neighborhood of $P'_{j_1,j_2}$, say $W'_{j_1,j_2} \subset
U'_{j_1,j_2}$, such that there exists an inverse of $f_{j_1,j_2}$
on $f_{j_1,j_2}(W'_{j_1,j_2})$, but this is not a neighborhood of
$P$. By Lemma \ref{nonnegativelincomb} we see that the
$f_{j_1,j_2}(W'_{j_1,j_2})$ cover a neighborhood of $P$ as
$(j_1,j_2)$ varies and we are done.
\end{proof}

We can now describe our strategy. Given a function
$F:\mathcal{M}\longrightarrow\mathbb{R}$ we can look for maxima
and minima of $F$ in following way:
\begin{enumerate}
\item study $F$ on rank one matrices using an ad hoc method. When
$F$ is the likelihood function, the problem is quite simple, see
e.g. \cite{agresti:02};

\item\label{step1} consider the functions $F\circ f_{j_1,j_2}$ and
compute their maxima and minima on $U'_{j_1,j_2}$ for all $1\leq
j_1<j_2 \leq J$ (notice that these computation are as simple as
they could be as the least number of variable is involved); let
$Q$ be one of the point we found;

\item if $Q$ is in the interior of one of the $U'_{j_1,j_2}$ then
$f_{j_1,j_2}(Q)$ is a maximum or minimum of $F$;

\item if $Q$ lies on the boundary of one of the $U'_{j_1,j_2}$ and
$f_{j_1,j_2}(Q)$ is on the boundary of $\mathcal{M}$, then
$f_{j_1,j_2}(Q)$ is a maximum or minimum of $F$;

\item if $Q$ lies on the boundary of one of the $U'_{j_1,j_2}$ and
$f_{j_1,j_2}(Q)$ has rank one we already treated this case in the
first step;

\item if $Q$ lies on the boundary of one of the $U'_{j_1,j_2}$ and
$f_{j_1,j_2}(Q)$ has two proportional columns, then $Q$ will lie
on the boundary of at least two of the $U'_{j_1,j_2}$; for each
each pair $(j_1,j_2)$ such that $Q$ is on the boundary of
$U'_{j_1,j_2}$ we have to compare the extremal behavior of the
functions $F\circ f_{j_1,j_2}$, if these behavior agree then
$f_{j_1,j_2}(Q)$ is a maximum/minimum of $F$ otherwise it is not.
\end{enumerate}

In this paper we only considered matrices of rank at most two. For
higher values of the rank the situation gets much more involved
and almost impossible to treat. For example, it is not even known
how to effectively compute the non-negative rank of a matrix. But,
some preliminary results in \cite{dong|lin|chu:09} suggest that
matrices with non-negative rank different from the ordinary rank
are exceptional, i.e. they form a zero-measure set. This
observation can be of some help to try and extend our approach.


\begin{thebibliography}{16}
\expandafter\ifx\csname
natexlab\endcsname\relax\def\natexlab#1{#1}\fi
\expandafter\ifx\csname url\endcsname\relax
  \def\url#1{{\tt #1}}\fi
\expandafter\ifx\csname
urlprefix\endcsname\relax\def\urlprefix{URL }\fi

\bibitem[{Agresti(2002)}]{agresti:02}
Agresti, A. (2002).
\newblock {\em Categorical Data Analysis\/}.
\newblock New York: Wiley, 2 ed.

\bibitem[{Catchpole and Morgan(1997)}]{catchpole|morgan:97}
Catchpole, E. and Morgan, B. (1997).
\newblock Detecting paramter redundancy.
\newblock {\em Biometrika\/}, {\em 84\/}(1), 187--196.

\bibitem[{Catchpole {\em et~al.\/}(1998)Catchpole, Morgan, and
  Freeman}]{catchpole|morgan|freeman:98}
Catchpole, E., Morgan, B., and Freeman, S. (1998).
\newblock Estimation in parameter-redundant models.
\newblock {\em Biometrika\/}, {\em 85\/}(2), 462--468.

\bibitem[{Catral {\em et~al.\/}(2004)Catral, Han, Neumann, and
  Plemmons}]{catral|han|neumann|plemmons:04}
Catral, M., Han, L., Neumann, M., and Plemmons, R. (2004).
\newblock On reduced rank nonnegative matrix factorization for symmetric
  nonnegative matrices.
\newblock {\em Linear Algebra Appl.\/}, {\em 393\/}, 107--126.

\bibitem[{Cohen and Rothblum(1993)}]{cohen|rothblum:93}
Cohen, J.~E. and Rothblum, U.~G. (1993).
\newblock Nonnegative ranks, decompositions, and factorizations of nonnegative
  matrices.
\newblock {\em Linear Algebra Appl.\/}, {\em 190\/}, 149--168.

\bibitem[{Dacunha-Castelle and Duflo(1986)}]{dacunhacastelle|duflo:86}
Dacunha-Castelle, D. and Duflo, M. (1986).
\newblock {\em Probability and statistics\/}.
\newblock New York: Springer Verlag.

\bibitem[{Dong {\em et~al.\/}(2009)Dong, Lin, and Chu}]{dong|lin|chu:09}
Dong, B., Lin, M.~M., and Chu, M.~T. (2009).
\newblock Nonnegative rank factorization via rank reduction.
\newblock Preprint.

\bibitem[{Fienberg {\em et~al.\/}(2010)Fienberg, Hersh, Rinaldo, and
  Zhou}]{fienberg|hersh|rinaldo|zhou:10}
Fienberg, S.~E., Hersh, P., Rinaldo, A., and Zhou, Y. (2010).
\newblock Maximum likelihood estimation in latent class models for contingency
  table data.
\newblock In P.~Gibilisco, E.~Riccomagno, M.~P. Rogantin, and H.~P. Wynn
  (Eds.), {\em Algebraic and Geometric Methods in Statistics\/}, Cambridge
  University Press. 27--62.

\bibitem[{Finesso and Spreij(2006)}]{finesso|spreij:06}
Finesso, L. and Spreij, P. (2006).
\newblock Nonnegative matrix factorization and $i$-divergence alternating
  minimization.
\newblock {\em Linear Algebra Appl.\/}, {\em 416\/}, 270--287.

\bibitem[{Gao {\em et~al.\/}(2008)Gao, Jiang, and Zhu}]{gao|jiang|zhu:08}
Gao, S., Jiang, G., and Zhu, M. (2008).
\newblock Solving the $100$ swiss francs problem.
\newblock Ar{X}iv:0809.4627v1.

\bibitem[{Harris(1992)}]{harris:92}
Harris, J. (1992).
\newblock {\em Algebraic geometry. A first course\/}, vol. 133 of {\em Graduate
  Texts in Mathematics\/}.
\newblock New York: Springer-Verlag.

\bibitem[{Ho and Van~Dooren(2008)}]{ho|vandooren:08}
Ho, N.-D. and Van~Dooren, P. (2008).
\newblock Non-negative matrix factorization with fixed row and column sums.
\newblock {\em Linear Algebra Appl.\/}, {\em 429\/}, 1020--1025.

\bibitem[{Lee and Seung(2000)}]{lee|seung:00}
Lee, D.~D. and Seung, H.~S. (2000).
\newblock Algorithms for non-negative matrix factorization.
\newblock In {\em NIPS\/}. 556--562.

\bibitem[{Pachter and Sturmfels(2005)}]{pachter|sturmfels:05}
Pachter, L. and Sturmfels, B. (2005).
\newblock {\em Algebraic statistics for computational biology\/}.
\newblock New York: Cambridge University Press.

\bibitem[{Pardo(2005)}]{pardo:05}
Pardo, L. (2005).
\newblock {\em Statistical Inference Based on Divergence Measures\/}.
\newblock Boca Raton: Chapman \& Hall/CRC.

\bibitem[{Sturmfels(2008)}]{sturmfels:08}
Sturmfels, B. (2008).
\newblock Open problems in algebraic statistics.
\newblock In M.~Putinar and S.~Sullivant (Eds.), {\em Emerging Applications of
  Algebraic Geometry\/}, New York: Springer, vol. 149 of {\em {I}.{M}.{A}.
  Volumes in Mathematics and its Applications\/}. 351--364.

\end{thebibliography}

\end{document}